\documentclass[aps,prd,reprint,twocolumn,tightenlines,eqsecnum,floats,amsmath,amssymb,nofootinbib
,superscriptaddress,showpacs]{revtex4-1}   

\usepackage{amsmath,amsthm,amssymb,amsfonts}
\usepackage{MnSymbol}
\usepackage{color}
\usepackage{colordvi}
\newtheorem{thr}{Theorem}
\newtheorem{prop}[thr]{Proposition}
\newtheorem{lm}[thr]{Lemma}
\newtheorem{df}{Definition}

\numberwithin{equation}{section}
\numberwithin{thr}{section}
\numberwithin{chr}{section}
\numberwithin{df}{section}

\newcommand{\scpr}[2]{\langle#1\, \vert \, #2 \rangle}

\newcommand{\bra}[1]{\langle\, #1\,\rvert}

\newcommand{\gbra}[1]{(\,#1\,\rvert}

\newcommand{\new}{\textrm{vtx}}
\newcommand{\ve}{\text{Vert}}

\DeclareMathOperator{\cyl}{Cyl}

\begin{document}
%%%%%%%%%%%%%%%%%%%%%%%%%%%%%%%%%%
\title{A symmetric scalar constraint for loop quantum gravity}

\author{Jerzy Lewandowski}\email{Jerzy.Lewandowski@fuw.edu.pl}
\affiliation{Faculty of Physics, Uniwersytet Warszawski,
ul. Pasteura 5, 02-093 Warszawa, Polska (Poland)}
\affiliation{Institute for Quantum Gravity, Department of Physics
Friedrich-Alexander-Universit\"at  Erlangen-N\"urnberg (FAU), Erlangen (Germany)}

\author{Hanno Sahlmann}
\email{hanno.sahlmann@gravity.fau.de}
\affiliation{Institute for Quantum Gravity, Department of Physics
Friedrich-Alexander-Universit\"at  Erlangen-N\"urnberg (FAU), Erlangen (Germany)}

\begin{abstract} 
In the framework of loop quantum gravity, we define a new Hilbert space of states which are solutions of 
a large number of components of the diffeomorphism constraint. On this Hilbert space, using the 
methods 
of Thiemann, we obtain a family of gravitational scalar constraints. They preserve the Hilbert space for every choice of lapse function. Thus 
adjointness and  commutator properties of the constraint can be investigated in a straightforward manner. We show how the space of solutions of the symmetrized constraint can be defined by spectral decomposition, and the Hilbert space of physical states by subsequently fully implementing the diffeomorphism constraint. 
The relationship of the solutions to those resulting from a proposal for a symmetric constraint operator by Thiemann remains to be elucidated.
\end{abstract} 
      
\maketitle

\newpage

\tableofcontents

%-------------------------------------------------------
\section{Introduction} 
\label{se_intro}
%------------------------------------------------------- 
In any known canonical formulation of general relativity, the general covariance of the theory is encoded in a number of constraints imposed on phase space. These constraints generate the hypersurface-deformation algebra under Poisson brackets, which is universal for generally covariant theories. 

For a canonical quantization of general relativity, it is thus vital that the constrains are implemented in the quantum theory. For the case of loop quantum gravity (LQG, see \cite{Ashtekar:2004eh,gr-qc/0509064} for a review), the diffeomorphism constraints have been dealt with successfully, resulting in a Hilbert space $\mathcal{H}_{\rm diff}$ off quantum states  that are invariant under spatial diffeomorphisms \cite{Ashtekar:1995zh}. The scalar constraints are technically much more demanding, as they have a more complicated action on the canonical variables employed in LQG. 

Thiemann \cite{gr-qc/9606089, gr-qc/9606090, gr-qc/9705017}, based in part on ideas by Rovelli and Smolin \cite{Rovelli:1993}  and  other researchers  \cite{AL:up}, and using the quantum volume operator \cite{AL:vol0}, succeeded in defining a quantum  scalar constraint $\widehat{C}(N)$. One essential ingredient was the introduction of a regulated version $\widehat{C}_\mathcal{R}(N)$ of the constraint on the kinematic Hilbert space $\mathcal{H}_{\rm kin}$, and the observation, that the regulator can be removed when the action of the regulated constraint is extended to $\mathcal{H}_{\rm diff}$ by duality. 
Due to the presence of the lapse function $N$,  the operator $\widehat{C}(N)$ is not invariant under spatial diffeomorphisms, and hence does not preserve $\mathcal{H}_{\rm diff}$. 
In fact, no Hilbert space which is invariant under $\widehat{C}(N)$ is known. This has turned out to be a substantial obstacle, as it precludes the straightforward discussion of adjointness relations, spectral resolutions, commutator algebra etc. of the $\widehat{C}(N)$'s. 
One way to deal with these difficulties is to work with the regulated constraints $\widehat{C}_\mathcal{R}(N)$ on $\mathcal{H}_{\rm kin}$. Thiemann showed in \cite{gr-qc/9606089} that the commutator of two regulated constraints vanishes. So does the extension of the commutator by the duality to $\mathcal{H}_{\rm diff}$.
It is in this sense that Thiemann's quantization is anomaly free. Remarkably, Thiemann was also able to devise a symmetric regulated constraint on $\mathcal{H}_{\rm kin}$ in \cite{gr-qc/9606090} which is  anomaly free in the same sense. A mathematically exact approach to the issue of commuting two  quantum scalar constraint operators  after removing the regulator was introduced in \cite{gr-qc/9710016} by extending the Hilbert space $\mathcal{H}_{\rm diff}$ to a suitable vector (no longer Hilbert) space  named a "habitat".  
Other ways to deal with the difficulties resulting from the unregulated constraint $\widehat{C}(N)$ not leaving its domain invariant have been suggested, see for example \cite{gr-qc/9705017,gr-qc/9710016}.

While these are acceptable resolutions to some of the problems, there remains some uneasiness due to the fact that  there is no scalar constrain operator without regulator, acting on, and leaving invariant, a Hilbert space. It is here that the present article contains a substantial improvement.
The solution we present has $\widehat{C}(N)$ act on -- and leave invariant -- a new Hilbert space $\mathcal{H}_{\new}$ of \emph{almost} diffeomorphism invariant states. These new states can be thought of as being obtained from the spin network states of LQG by averaging over their images under diffeomorphisms which leave fixed sets of vertices in the spatial manifold invariant, schematically, 
\begin{equation}
\label{eq_ga}
\overline{\Psi_\gamma}\propto\sum_{f \in \text{Diff}(\Sigma)_{v_1\ldots v_n}}U_f \Psi_\gamma
\end{equation}
where $\gamma$ is a graph with vertices $v_1\ldots v_n$, the sum is over elements of the stabilizer of the vertex set, and $U$ is the unitary 
action of the diffeomorphisms.  Thiemann \cite{gr-qc/9606089} defines a regulated operator $\widehat{C}_\mathcal{R}(N)$ and then shows that 
the limit $\mathcal{R}\rightarrow 0$ is well defined in a relatively weak sense, namely on diffeomorphism invariant distributions. Technically, the regulator is similar to that in lattice gauge theory -- curvature is approximated by traces of holonomies. When acting with the regulated operator on the new Hilbert space $\mathcal{H}_{\new}$, the partial group averaging \eqref{eq_ga} is enough to obtain a well defined operator in the limit of  vanishing regulator. On the other hand, mostly due to the nature of spatial volume in loop quantum gravity, the resulting state will still belong to $\mathcal{H}_{\new}$. 

In the new Hilbert space, adjointness and commutator properties of the constraint can be investigated, and a physical Hilbert space defined by using spectral decomposition.  

It is a very interesting -- and open -- question, how solutions to the new scalar constraint on 
$\mathcal{H}_{\new}$ relate to those of Thiemann's symmetric scalar constraint \cite{gr-qc/9606090}. Similarly, one could symmetrize Thiemann's non-symmetric scalar constraint on $\mathcal{H}_{\rm diff}$ for the case of \emph{constant} lapse function and compare to our proposal.

We should note that there is a very interesting different line of thought, \cite{Laddha:2010wp,Laddha:2011mk,Tomlin:2012qz,Varadarajan:2012re}, which also suggests that one should use a different Hilbert space to represent the (diffeomorphism and scalar) constraints. Those methods carry the additional benefit that they address the question of anomalies in a direct fashion. 

The present article is organized as follows. In section \ref{se_kin}, we briefly recall the setup of kinematic quantization in loop quantum gravity. Section \ref{se_new} introduces the new Hilbert space, which is used in sections \ref{se_euc} and \ref{se_lor}, respectively, for the quantization of Euclidean and Lorentzian scalar constraints.  The space of solutions to all constraints is discussed in section \ref{se_sol}. A summary can be found in section \ref{se_sum}. 
%-------------------------------------------------------
\section{Kinematic quantization} 
\label{se_kin}
%-------------------------------------------------------
In this section, for completeness and to fix notation we briefly recall some basic notions of loop quantum gravity. 

%---------------------------------------------------
\subsection{Classical theory}
%----------------------------------------------------
In this article, we will consider 4d Einstein gravity in vacuum, given by the action
\begin{equation*}S[\phi,e,\omega]= S_\text{GR}+S_\text{Holst}\end{equation*}
with 
\begin{align}
S_\text{GR}&=\frac{1}{32\pi G}\int \epsilon_{IJKL} e^I\wedge e^J\wedge F^{KL}(\omega)\\
S_\text{Holst}&=-\frac{1}{16\beta G}\int e^I\wedge e^J\wedge F_{IJ}(\omega)
\end{align}
The canonical analysis of this action, and a partial gauge fixing (time gauge) leads to a phase space 
$\Gamma$  for the gravitational field. For a detailed derivation see for example \cite{gr-qc/0110034, Ashtekar:2004eh}. It is coordinatized by the su(2) valued                              
1-form field
\begin{equation*} A(x)\ =\  A_a^i(x) \tau_i \otimes dx^a,\  \end{equation*} 
and the canonically conjugate momentum vector-density
\begin{equation*} E(x)\ =\  E^a_i(x)\tau^{*i}\otimes \partial_a,\ \end{equation*}
taking values on a spatial slice $\Sigma$ of space-time.
 Indices $a,b,\ldots$ are spatial, whereas $i,j,\ldots$ refer to su(2), the algebra of the gauge group after partial gauge fixing. The usual choice of the basis $\tau_1,\tau_2,\tau_3$ is such that
\begin{equation*} 
[\tau_i,\tau_j]=\sum_k \epsilon_{ijk}\tau_{k}.
\end{equation*}
$\tau^{*i}$ denotes the dual basis in su(2)$^*$.

The Poisson bracket between two functionals $F[A,E],\ G[A,E]$ is
\begin{equation*} \{F,G\} = 8\pi G \beta
\int_\Sigma d^3x \frac{\delta F}{\delta A^i_a(x)} \frac{\delta G}{\delta E^a_i(x)}
- \frac{\delta G}{\delta A^i_a(x)} \frac{\delta F}{\delta E^a_i(x)} .
\end{equation*}
The phase space   $\Gamma$ is not yet physical, however. Rather, the physical phase space is induced by constraints on $\Gamma$. The main concern of the present work is the implementation in the quantum theory of the scalar constraint
\begin{equation}\label{vacuum_scalar}
C = \sqrt{\frac{\beta}{8\pi G}}  
\frac{E^a_iE_j^b}{\sqrt{|\det E|}}\left(\epsilon^{ijk}F_{abk}+2(\sigma-\beta^2)K_{[a}^iK_{b]}^j \right)
\end{equation}
is the scalar constraint of vacuum gravity. $F$ is the curvature of $A$ and $K$ is the extrinsic curvature of $\Sigma$, which is a function of $A$ and $E$.  For the Lorentzian gravity $\sigma=-1$. The Euclidean model of gravity is
defined by $\sigma=1$.  
%-----------------------------------------------------------------------
\subsection{Kinematic Hilbert space}
%-----------------------------------------------------------------------
In the present section, we will quantize the kinematic phase space $\Gamma$, resulting in a Hilbert space ${\cal H}$. The quantum states in LQG are \emph{cylindrical} functions of the variable $A$, i.e., they depend on $A$ only through finitely many holonomies 
\begin{equation} h_e[A]\ =\ \rm{Pexp}\left(-\int_e A\right) \end{equation}
where $e$ ranges over finite curves -- we will also refer to them as  \emph{edges} -- in $\Sigma$. 

To spell out the definition we need to be precise about the meaning of ``embedded graph'' used in the definition of the cylindrical function. A graph $\gamma$ embedded in $\Sigma$ is a set of edges (un-oriented) embedded in $\Sigma$, 
$\gamma=\{e_1,\ldots,e_n\}$, of three types:
 
\begin{enumerate}
\item  embedded  closed interval   (two end points),   
\item immersed   interval, such that the endpoints of the image coincide,  and there is no more
self-intersections (one endpoint),
\item embedded  circle (no endpoints). 
\end{enumerate} 
The end points of the edges of a given graph $\gamma$ form the set $\{v_1,\ldots ,v_m\}$ of the vertices 
of $\gamma$. Intersection of two different edges is either  empty or consists of vertices of $\gamma$, 
\begin{equation*}e_I\cap e_J \subset \{v_1,\ldots ,v_m\}, \ \ {\rm whenever}\ \ I\not= J. \end{equation*}
In particular, each edge of the  type 3 (circle) does not intersect any other
edge of $\gamma$. 

\begin{df}
A function  $\Psi:A\mapsto\Psi[A]$ is called \emph{cylindrical} if there is a graph $\gamma$ such that
\begin{equation}\label{cyl'} 
\Psi[A]\ =\ \psi(h_{e_1}[A],\ldots ,h_{e_n}[A]) 
\end{equation} 
with a function
function $\psi\ :\ {\rm SU}(2)^n \ \rightarrow \mathbb{C}$. Here, for every edge we choose an orientation 
to define the parallel transport $h_{e_I}[A]$. For each edge $e_J$ of the  type 3, we also choose an arbitrary beginning-end point, and assume that 
$$\psi(h_1,\ldots ,h_J,\ldots) = \psi(h_1,\ldots ,g^{-1}h_Jg,\ldots)\quad \forall g\in {\rm SU}(2).$$ 
\end{df}
Some remarks about this definition are in order: Firstly, we understand \eqref{cyl'} to include the case of $n=0$, in which case $\Psi[A]=$const. Furthermore, 
the functions $\psi$ in (\ref{cyl'}) can be arbitrary, however they must be restricted either to $L^2$ functions when we integrate them (to calculate the scalar product),
or to $C^n$ when we differentiate  them (to define quantum operators). The safe choice is to  assume that  $\psi$ is a polynomial in $\rho_1(h_1)$, \ldots , $\rho(h_n)$, where $\rho_I$ are representations of SU$(2)$ including the trivial one.

To calculate the scalar product between two cylindrical functions $\Psi$ and $\Psi'$
defined by using graphs $\gamma$ and  $\gamma'$, respectively, we find a refined 
graph $\gamma''=\{e''_1,\ldots ,e''_{n''}\}$, such that  both the functions can be written as \footnote{The existence 
is ensured by assuming a suitable differentiability class of the edges. A safe assumption is analyticity of the edges. 
Since analytic diffeomorphisms are not local enough,
in \cite{LOST} we introduced a new category of manifolds we called \emph{semianalytic}. Briefly, semianalyticity means 
differentiability of a given finite order, and suitably defined piecewise analyticity. Then, all the edges and diffeomorphisms are assumed to be semianalytic.}
\begin{equation*}
\begin{split}
\Psi[A] &= \psi(h_{e''_1}[A],\ldots ,h_{e''_{n''}}[A]),\\
\Psi'[A] &= \psi'(h_{e''_1}[A],\ldots ,h_{e''_{n''}}[A]).
\end{split}
\end{equation*}
The scalar product is
\begin{equation}\label{(|)} \scpr{\Psi}{\Psi'}\ =\ \int dg_1\ldots dg_{n"}\overline{\psi(g_1,\ldots ,g_{n"})}{\psi}'(g_1,\ldots ,g_{n"}).\end{equation}
We denote the space of all the cylindrical functions defined as above with a graph $\gamma$ by $\tilde{\cyl}_\gamma$ and, 
respectively, the space of all the cylindrical functions by $\cyl$. The Hilbert space ${\cal H}_{\rm kin}$ is the completion 
\begin{equation*} {\cal H}_{\rm kin} \ =\ \overline{\cyl}\end{equation*}
with respect to the Hilbert norm defined by (\ref{(|)}).

Every cylindrical function $\Psi$ is also a quantum operator
\begin{equation}(\Psi(\widehat{A})\Psi')[A]\ =\ \Psi[A]\Psi'[A].\end{equation}  
A connection operator $\widehat{A}$ by itself is not defined. 

The field $E$ is naturally quantized as 
\begin{equation}\widehat{E}^a_i\Psi[A]\ =\ \frac{\hbar}{i}\{\Psi[A],E^a_i(x)\} \ =\ \frac{8\pi\beta l_\text{P}^2}{i}
\frac{\delta}{\delta A^i_a}\Psi[A]. 
\end{equation}
It turns into well defined operators in ${\cal H}_{\rm kin}$ after smearing along a 2-surface   $S\subset \Sigma$ 
$$\int_S \frac{1}{2}f^i\widehat{E}^a_i\epsilon_{abc}dx^b\wedge dx^c\ \ \ \ \  f:S\rightarrow {\rm su}(2)$$
where $f$ may involve parallel transports \cite{Rovelli:1989za,Thiemann:2000bv}: 
$$ f(x)\ =\ \tilde{f}(x)\left(h_{p{x_0x}}\xi h_{p{xx_0}}\right)^i     $$
where $S\ni x\mapsto p_{xx_0}$ assigns to each point $x\in S$ a path $p_{xx_0}$ connecting a fixed point $x_0$ to $x$,
$\xi\in$su(2), and $\tilde{f}:S\rightarrow \mathbb{R}$. 

There is an orthogonal decomposition of  ${\cal H}_{\rm kin}$ with respect to subspaces labeled by 
the graphs defined above. 
Given a graph $\gamma$, denote by $\widetilde{{\cal H}}_{\gamma}$ the subspace of ${\cal H}_{\rm kin}$ defined by the 
cylindrical functions  (\ref{cyl'}) corresponding to  $\gamma$ . Whenever a graph $\gamma$ can be obtained from a graph $\gamma'$ 
by a sequence of the following steps:
\begin{itemize}
\item cutting an edge $e'_I$ into two: $e'_I=e_{J}\circ e_{K}$
\item adding a new edge: $\gamma=\{e'_1,\ldots ,e'_{n-1},e_n\}$, $\gamma'=\{e'_1,\ldots ,e'_{n-1}\}$
\end{itemize}
then 
\begin{equation*} \tilde{{\cal H}}_{\gamma'}\ <\  \tilde{{\cal H}}_{\gamma}, \end{equation*}
that is, $ \tilde{{\cal H}}_{\gamma'}$  is a proper subset of $ \tilde{{\cal H}}_{\gamma}$.   Hence,
\begin{equation*} {\cal H}_{\rm kin}\ =\ \overline{\bigcup_{\gamma}\tilde{{\cal H}}_\gamma},  \end{equation*}
but this is not an orthogonal decomposition.  

Define $\Psi\in \tilde{{\cal H}}_{\gamma}$ to be a proper element of $\tilde{{\cal H}}_{\gamma}$ 
if this is true that  
\begin{equation*}  \Psi\, \perp\,  \tilde{{\cal H}}_{\gamma'}\ \ \Leftarrow\ \ \tilde{{\cal H}}_{\gamma'}\ <\  \tilde{{\cal H}}_{\gamma}   . \end{equation*}
Given $\gamma$, the proper states form a subspace ${\cal H}_{\gamma}\ \subset\ \tilde{{\cal H}}_\gamma$. The family $({\cal H}_{\gamma})_\gamma$ does provide an orthogonal decomposition
\begin{equation}\label{decomp}
{\cal H} _{\rm kin}\ =\ \overline{\bigoplus_{\gamma} {\cal H}_{\gamma}}.
\end{equation}
This decomposition can also be applied directly to the cylindrical functions 
\begin{align} \cyl_\gamma\ &:=\ \cyl\cap {\cal H}_\gamma \ \subset \cyl,\\
                     \cyl\ &=\ \bigoplus_\gamma \cyl_\gamma \end{align}

%-----------------------------------------------------------------------------
\section{Quantum scalar constraint}
%-----------------------------------------------------------------------------
The scalar constraint
$$C(N)\ =\ \int d^3x N(x) C(x)$$ 
has not been successfully  quantized in the kinematical Hilbert space ${\cal H}_{\rm kin}$ of the previous section.   We will introduce now a new Hilbert space which admits  quantum operators 
$\widehat{C}(N)$. 
%-----------------------------------------------------------------------------
\subsection{A new Hilbert space} 
\label{se_new}
%-----------------------------------------------------------------------------
The idea of the new Hilbert space we will introduce now is to average each of the subspaces 
${\cal H}_{\gamma}$ with respect to all the diffeomorphisms  Diff($\Sigma$)$_{v_1,\ldots ,v_m}$  
which act trivially on the set  ${v_1,\ldots ,v_m}$ of the vertices of $\gamma$.\footnote{The general idea of the 
averaging with respect to the diffeomorphisms has been already used in LQG, see \cite{Ashtekar:1995zh,Ashtekar:2004eh}. We apply it now for our purposes.} 

Every $f\in $ Diff($\Sigma$) defines a unitary operator $U_{f}:{\cal H}_{\rm kin}\rightarrow {\cal H}_{\rm kin}$,
\begin{equation*} 
U_f\Psi[A]\ =\ \Psi[f^*A]. 
\end{equation*}
Given a graph $\gamma$  consisting of edges and vertices
\begin{equation*}
E(\gamma):=\{e_1,\ldots ,e_n\},\quad
\ve(\gamma)=\{v_1,\ldots ,v_m\}, 
\end{equation*}
the action of $U_f$ on a cylindrical function (\ref{cyl'}) reads
\begin{equation} 
U_f\Psi[A]\ =\ \psi(h_{f(e_1)}[A],\ldots ,h_{f(e_n)}[A]),
 \end{equation}
 where for the parallel transport along each edge $f(e_I)$ we choose the orientation 
 induced by the map $f$ and by the orientation of $e_I$ chosen in (\ref{cyl'}).
Denote by TDiff($\Sigma$)$_\gamma$ the subset of Diff($\Sigma$) which consists of all the diffeomorphisms acting trivially in  $\tilde{{\cal H}}_{\gamma}$.  On the other hand, for a general $f\in {\rm Diff}(\Sigma)$, we have a unitary isomorphism
\begin{equation*} 
U_f: {\cal H}_\gamma\ \longrightarrow {\cal H}_{f(\gamma)}.
\end{equation*}
The maps 
${\cal H}_{\gamma}\longrightarrow {\cal H}_{\rm kin}$
obtained by the diffeomorphisms Diff($\Sigma$)$_{v_1,\ldots ,v_m}$ are in a one-to-one  correspondence with  the 
elements of the quotient 
\begin{equation} 
{\rm D}_\gamma := {\rm Diff}(\Sigma)_{\ve(\gamma)} / {\rm TDiff}(\Sigma)_\gamma. 
\end{equation}
Still, ${\rm D}_\gamma$ is a noncompact set and we do not know any probability measure on it.   
Therefore it is not surprising, that given $\Psi \in {\cal H}_\gamma$,  the result of the averaging will not, in general, be an
element of  ${\cal H}_{\rm kin}$. However, it will be  well defined as an element of
the space $\cyl^*$ dual to $\cyl$.
Given $\Psi \in {\cal H}_\gamma$, we turn it into  $\gbra{\Psi} \in \cyl^*$,
$$ \gbra{\Psi} \ :\ \Psi'\mapsto \scpr{\Psi}{\Psi'},$$
 and average in $\cyl^*$,
\begin{equation}\label{eta}
 \eta(\Psi)\ =\ \frac{1}{N_\gamma} \sum_{[f]\in D_\gamma}
\gbra{U_f\Psi}, 
\end{equation} 
$N_\gamma$ is a normalization factor that will be determined in a moment. 
\begin{lm}
$\eta(\Psi)$ is a well-defined linear functional 
\begin{equation*}
\eta(\Psi): {\rm Cyl}\rightarrow \mathbb{C}
\end{equation*}
which is invariant under ${\rm Diff}(\Sigma)_{\ve(\gamma)}$. 
\end{lm}
\begin{proof}
Each term in the sum (\ref{eta}) is independent of the choice of a representative $f\in[f]$
 because the action of ${\rm TDiff}(\Sigma)_\gamma$ on $\mathcal{H}_\gamma$ is trivial. 
 Given $\Psi'\in\cyl$, only a finite set of terms in the sum is not zero. Hence, the sum is finite.  
 The sum is invariant, because if $\{f_i\}$ is a set of representatives for the classes $D_\gamma$, then so is $\{f_0f_i\}, f_0\in {\rm Diff}(\Sigma)_{\ve(\gamma)}$. 
\end{proof}
We define the map
\begin{equation*} {\cal H}_\gamma \ni \Psi\ \mapsto\  \eta(\Psi)\ \in\ {\rm Cyl}^*\end{equation*} 
for every embedded graph $\gamma$, and extend it by linearity to the algebraic orthogonal sum 
\begin{equation}
\eta:\bigoplus_\gamma {\cal H}_\gamma \longrightarrow \cyl^*.
\end{equation}
Notice, that 
$ \cyl \subset \bigoplus_\gamma {\cal H}_\gamma$,
therefore $\cyl$ is in the domain of the averaging map $\eta$.  
\begin{df}
The new Hilbert space  ${\cal H}_{\new}$ is defined as the completion
\begin{equation}{\cal H}_{\new}\ =\ \overline{\eta({\rm Cyl})}  
\end{equation}
under the norm induced by the scalar product
\begin{equation} \label{scpr}
(\eta(\Psi)|\eta(\Psi'))\ :=\ \eta(\Psi)(\Psi'). 
\end{equation} 
\end{df}
One can check \cite{Ashtekar:1995zh} that  \eqref{scpr} has indeed all the properties of a scalar product, 
and hence ${\cal H}_{\new}$ really is a Hilbert space. It has an orthogonal decomposition that is reminiscent of 
\eqref{decomp}:
\begin{lm}
Let ${\rm FS}(\Sigma)$ be the set of finite subsets of $\Sigma$. Then
\begin{align} \label{decomp1}
{\cal H}_{\new}\ &=\ \overline{\bigoplus_{V\in {\rm FS}(\Sigma)}{\cal H}_{V}}\\
\label{decomp2}
{\cal H}_{V}\ &=\ \overline{\bigoplus_{[\gamma]\in[\gamma(V)]}\mathcal{H}_{[\gamma]}}\\
\mathcal{H}_{[\gamma]}\ &= \ \eta({\cal H}_{[\gamma]}),
\end{align}
where  $\gamma(V)$ is the set of graphs $\gamma$ with vertex set $V=\ve(\gamma)$, $[\gamma(V)]$ is the set of the 
Diff$(\Sigma)_V$-equivalence classes $[\gamma]$ of the graphs $\gamma\in \gamma(V)$.  
\end{lm}
\begin{proof} Both decompositions follow from definition (no spurious vertices) and the orthogonality of the subspaces $\mathcal{H}_\gamma$, together with \eqref{eta}. 
\end{proof}
To understand the structure of  each of the subspaces $\eta({\cal H}_\gamma)$, 
decompose the space ${\cal H}_{\gamma}$  into the kernel of $\eta$, and the orthogonal  completion 
\begin{equation} 
{\cal H}_{\gamma}\ =\ {\rm Ker}(\eta)\cap {\cal H}_{\gamma}\ \oplus\ S_\gamma.
\end{equation}
The orthogonal completion $S_\gamma$ consists of all the $\Psi$ such that for every $f\in{\rm Diff}(\Sigma)_{\ve(\gamma)}$
\begin{equation}
 f(\gamma)=\gamma\ \Longrightarrow\  U_f\Psi\ =\ \Psi.
\end{equation}
In other words, elements of $S_\gamma$ are invariant with respect to the symmetry group 
\begin{equation} {\rm Sym}_\gamma\ =\ \{f\in  {\rm Diff}(\Sigma)_{\{x_1,\ldots ,x_m\}}\ :\ f(\gamma)=\gamma\}\, /\, {\rm TDiff}(\Sigma)_\gamma
\end{equation} 
of the graph $\gamma$. In fact, it is straightforward to show the following
\begin{lm} 
The map
\begin{equation*} 
\eta: S_\gamma\rightarrow \eta({\cal H}_\gamma)
\end{equation*}
is a unitary  embedding modulo an overall factor  
${|{\rm Sym}_\gamma|}/{N_\gamma}$, where $N_\gamma$ is the free constant in the definition \eqref{eta} of $\eta$. 
\end{lm}
In the following, we set 
$$\frac{|{\rm Sym}_\gamma|}{N_\gamma}\ =\ 1.$$

Finally, we point out that ${\cal H}_{\new}$ carries a natural action of Diff$(\Sigma)$, which we will also denote by $U$. It is defined by
\begin{equation}
\label{diff}
U_f \eta(\Psi):=\eta(U_f \Psi), \qquad f\in {\rm Diff}(\Sigma)
\end{equation}
and extension by continuity. A short calculation shows
\begin{lm}
$U_f$ as in \eqref{diff} is unitary and maps $\mathcal{H}_V$ to  $\mathcal{H}_{f(V)}$ in the decomposition \eqref{decomp1}. 
\end{lm}

%-----------------------------------------------------------------------------
\subsection{Lifting operators to $\mathcal{H}_{\new}$}
%---------------------------------------------------------------------------
In the kinematical Hilbert space ${\cal H}$ one often considers quantum operators defined 
on the domain $\cyl$, and such that 
\begin{equation} \label{O}\widehat{\cal O}\ : {\rm Cyl} \rightarrow {\rm Cyl} . \end{equation}
Each of them passes, by duality, to an operator $\widehat{\cal O}^*$ defined in Cyl$^*$. In particular, it is defined on 
$\eta(\cyl)\subset {\cal H}_{\new}$.  However,  while $\widehat{\cal O}^*$
maps  $\eta(\cyl)$ into Cyl$^*$, the image is not  necessarily in ${\cal H}_{\new}$. 
Importantly, sometimes the domain is actually  mapped  back into ${\cal H}_{\new}$. Then  $\widehat{O}^*$ becomes 
an operator in ${\cal H}_{\new}$. 
\begin{lm} 
Suppose, an operator as in \eqref{O} 
is of the form 
\begin{equation*}
\widehat{O}(N)\ =\ \sum_{x\in\Sigma} N(x) \widehat{\cal O}_x,  
\end{equation*} 
where the $\widehat{\cal O}_x$ are operators that have a \emph{local action} for all $x\in\Sigma$, 
\begin{align}  
\label{loc1}   
U_f \widehat{O}_x  &= \widehat{O}_x U_f \quad \text{ for }f\in {\rm Diff}_{\{x\}}\\
\label{loc2}
\widehat{\cal O}_x|_{\tilde{\cal H}_\gamma} &= 0 \quad \text{ for }x \notin \ve(\gamma) .
\end{align}
Then  $\widehat{\cal O}^*$ is an operator on $\mathcal{H}_{\new}$.
\end{lm}
\begin{proof}
Using the conditions of locality (\ref{loc1},\ref{loc2}), one can pull the action of $\widehat{O}$ through the average \eqref{eta}, resulting in an element of Cyl$^*$ of the same form.
\end{proof}
There are several important operators which have this property: the quantum volume element smeared against
an arbitrary function $N\in C(\Sigma)$,
\begin{equation*}
 \widehat{V}(N)\ =\ \int d^3x N(x)  \sqrt{|{\rm det}\widehat{E}|}\ =\ \sum_{x\in\Sigma}N(x)\widehat{V}_x \end{equation*}
the Gauss constraint operator (for $\Lambda\in C(\Sigma,{\rm su(2)})$)
\begin{equation*}
 \int d^3 x \Lambda^i D_a\widehat{E}^a_i\ =\ \sum_{x\in\Sigma}\Lambda^i(x)\widehat{\cal G}_{ix} ,\end{equation*}
and also the integral of the Ricci scalar operator
\begin{equation*}
\widehat{R}(N)\ =\ \int d^3x N(x)   \widehat{{\sqrt{|{\rm det}{E}|}R}}(x)\ =\ \sum_{x\in\Sigma}N(x)\widehat{VR}_x .
\end{equation*}
which has recently been introduced \cite{Alesci:2014aza}. Our quantum scalar constraint operator will take a similar form in ${\cal H}_{\new}$, although
it will not be well defined in Cyl itself. It will be defined directly in ${\cal H}_{\new}$. 

%-----------------------------------------------------------------------------
\subsection{Scalar constraint operator for Euclidean gravity} 
\label{se_euc}
%-----------------------------------------------------------------------------
The Euclidean scalar constraint in the absence of matter can be obtained from \eqref{vacuum_scalar} by setting the metric signature $\sigma$ to 1. For the choice of $\beta=1$, the expression simplifies because the second term drops out. The remaining term, in Thiemann form, is proportional to 

\begin{equation}\label{euclidean_scalar}
C_{\rm Euc}(N) = \frac{-2}{(8\pi G)^2\beta^{\frac{3}{2}}} \int d^3 x \epsilon^{abc}{\rm Tr}F_{ab}(x)\{A_b(x), V(N)\} \end{equation} 
where 
$$V(N)= 
\int d^3x N(x)\sqrt{{\rm det}\ E(x)} $$    
and $N: \Sigma \rightarrow \mathbb{R}$ is an arbitrary lapse function. 

In the Lorentzian case, the second term in \eqref{vacuum_scalar} cannot be made to vanish for real Ashtekar variables, 
thus $C_{\rm Euc}(N)$ is only one part of $C(N)$. But even in this case, 
$C_{\rm Euc}(N)$ plays a vital part in the quantization of the whole constraint \cite{gr-qc/9606089}. 

To quantize the Euclidean scalar constraint, we express  $F$ by the parallel transports along suitable 
loops  $\alpha_\sigma^{\epsilon}$ and we express  $A$ in terms of parallel transport along suitable curves 
$s_\sigma^{\epsilon}$, 
\begin{equation} 
\begin{split}
C^{\epsilon}_{\rm Euc}(N)(A,E)\ =\  \sum_\sigma  & B_\sigma 
\left( \rho_\sigma(h_{\alpha_\sigma^{\epsilon}}) - \rho_\sigma(h_{(\alpha_\sigma^{\epsilon})^{-1}})\right)\cdot\\ 
& \cdot {\rm Tr}\left(\rho_\sigma(h_{(s_\sigma^{\epsilon})^{-1}}) \{\rho_\sigma(h_{s_\sigma^{\epsilon}}), 
V(N)\}\right)
 \end{split}
\end{equation}
The $B_\sigma$ are $\epsilon$-independent constants and the $\rho_\sigma$  representations of SU$(2)$. The loops $\alpha$ and  curves $s$  approach points in the limit $\epsilon\rightarrow 0$. Moreover,  $B_\sigma$, $\rho_\sigma$, $\alpha_\sigma$ and $s_\sigma$ are chosen such that, 
the entire expression converges to \eqref{euclidean_scalar}, 
\begin{equation} 
\lim_{\epsilon\rightarrow 0}C^{\epsilon}_{\rm Euc}(N)(A,E) = C_{\rm Euc}(N)(A,E)
\end{equation}
for every smooth $(A,E)$.   
   
For every fixed value of $\epsilon$, the operator
\begin{equation}
 \begin{split}
 \widehat{C}^{\epsilon}_{\rm Euc}(N) &=
 \frac{1}{i\hbar}\sum_\sigma   B_\sigma 
 {\rm Tr}
 \Big( \left(\widehat{\rho_\sigma(h_{\alpha_\sigma^\epsilon})}-\widehat{\rho_\sigma(h_{(\alpha_\sigma^\epsilon)^{-1}})}\right)\cdot\\
 &\qquad\qquad\cdot \widehat{\rho_\sigma(h_{(s_\sigma^\epsilon})^{-1})} [\widehat{\rho_\sigma(h_{s_\sigma^\epsilon})}, \widehat{V}(N)]\Big)  ,
 \end{split}
\end{equation}
 is well defined in the kinematic Hilbert space ${\cal H}_{\rm kin}$ in the domain $\cyl$. 
 However, the limit $\epsilon\rightarrow 0$
 does not exist.  Also, before taking the limit, for a constant  $\epsilon$, the operator is  not diffeomorphism covariant. 
 The finite loops break the covariance.  Remarkably, there is a way out.  First, we improve the regularization. 
For that we apply the decomposition  
(\ref{decomp}), and adapt the regulated expression to each subspace ${\cal H}_\gamma$ independently. 
We will do it below in such a way,  that for $\gbra{\Psi_1} \in \eta(\cyl)\subset {\cal H}_{\new}$  
and $\Psi_2\ \in \cyl_\gamma\subset  {\cal H}_{\gamma}$, the number  
\begin{equation}  \label{<epsilon>} 
\gbra{\Psi_1} \left(\widehat{C}^{\epsilon}_{\rm Euc}(N) \Psi_2 \right)
\end{equation}                                     
will be $\epsilon$-independent -- either because it vanishes, or due to the symmetries of $\gbra{\Psi_1}\ \in\ {\cal H}_{\new}$. In this way, we will define the limit 
\begin{equation}
\widehat{C}^*_{\rm Euc}(N):= \lim_{\epsilon\rightarrow 0} \left(\widehat{C}^{\epsilon}_{\rm Euc}(N)\right)^*
\end{equation}
as an operator on $\cyl^*$, by setting
\begin{equation}  \label{lim}
\left( \widehat{C}^*_{\rm Euc}(N)   \gbra{\Psi_1}  \right)(\Psi_2)
:= \lim_{\epsilon\rightarrow 0} 
\gbra{\Psi_1} \left(\widehat{C}^{\epsilon}_{\rm Euc}(N) \Psi_2 \right).
\end{equation}  
Note that this involves some abuse of notation, as $\widehat{C}^*_{\rm Euc}$ is not 
the dual of any operator defined in $\cyl$. 

We explain now, in what way we achieve the $\epsilon$-independence of (\ref{<epsilon>}).
We are making the same assumptions about the loop-path assignment  
\begin{equation}\label{la}
(\gamma, v) \mapsto \{\alpha^\epsilon_\sigma, s^\epsilon_\sigma\, |\, \sigma=1,2,\ldots \}
\end{equation}
as in Sec.\ VI.C of \cite{Ashtekar:2004eh}. 
For each $\sigma$, the pair $s_\sigma$ and $\alpha_\sigma$ is based at a point $v\in\Sigma$. 
If $v$  is not a vertex of $\gamma$, then the corresponding   term of the operator vanishes.
Here, $v\in \ve(\gamma)$, and $\sigma$ labels the pairs of 
loops and segments based at $v$. 
As a result, the action of the regulator-dependent operator $\widehat{C}^{\epsilon}_{\rm Euc}(N)$ defined on ${\cal H}_\gamma$ takes the following form,
\begin{equation}\label{opform}
\widehat{C}^{\epsilon}_{\rm Euc}(N) =\ \sum_{v\in \ve(\gamma)} N(v)\sum_{\sigma}\widehat{C}^\epsilon_{\gamma v \sigma},
\end{equation}
where 
\begin{equation}\label{epsilon}  
\widehat{C}^\epsilon_{\gamma v\sigma}: {\cyl}_{\gamma}  \rightarrow \cyl_{\gamma\cup \{\alpha^\epsilon_\sigma\}}\ \subset\ {\cal H}_{\gamma\cup \{\alpha^\epsilon_\sigma\}}.
\end{equation} 
In other words, the operator adds the loops $\alpha^\epsilon_\sigma$ to $\gamma$, while the paths $s^\epsilon_\sigma$, 
possibly new elements of the  graph $\gamma$,  do not change the graph, regarded as a subset in $\Sigma$. 
Every loop  $\alpha^\epsilon_\sigma$  appearing in (\ref{epsilon}) begins and ends at a vertex of $\gamma$, does not
intersect $\gamma$ in any other point, and does not have self intersections. Hence   $\alpha^\epsilon_\sigma$
becomes an edge (of type 2) of the new graph $\gamma\cup \{\alpha^\epsilon_\sigma\}$.  
One of the two key properties we ask of the loop 
assignment \eqref{la} is that for every $\epsilon_1$ and $\epsilon_2$ there is $f\in {\rm Diff}_{\ve(\gamma)}$ such that
\begin{equation}\label{cov}
\widehat{C}^{\epsilon_2}_{\gamma v \sigma}\ =\ U_f \widehat{C}^{\epsilon_1}_{\gamma v \sigma} . \end{equation}
This is the property that ensures the independence of (\ref{<epsilon>}) of $\epsilon$: 
\begin{equation}  \label{<epsilon12>} \begin{split}
\bra{\Psi_1} \left(\widehat{C}^{\epsilon_2}_{\gamma v\sigma}  \Psi_2 \right) &= \bra{\Psi_1} \left(U_f\widehat{C}^{\epsilon_1}_{\gamma v\sigma}  \Psi_2 \right)\\
&= \bra{\Psi_1} \left(\widehat{C}^{\epsilon_1}_{\gamma v_I\sigma}  \Psi_2 \right).
\end{split}
\end{equation}   
Consequently, the limit \eqref{lim} on $\eta(\cal{H}_\gamma)$ and, by linearity over the $\gamma$-sectors, on $\eta(\cyl)\ \subset\ {\cal H}_{\new}$, can be taken. 
However, the result  is not necessarily an element of ${\cal H}_{\new}$. 
For example, in general it may not be Diff$(\Sigma)_{\ve(\gamma)}$-invariant.  To ensure the invariance, we need to  coordinate the assignments \eqref{la} for 
graphs that are equivalent under Diff$(\Sigma)_{v_1,\ldots ,v_m}$. 
Also, we want the resulting operator valued distribution to be invariant with respect to all the Diff$(\Sigma)$. 
Hence, as our second key property, we ask the following: 
Given a graph $\gamma$, and 
 $f_1\in{\rm Diff}(\Sigma)$ then there exists 
 $f_2\in {\rm Diff}(\Sigma)$
such that
\begin{equation} \label{CffC}
 \widehat{C}^{\epsilon}_{f_1(\gamma)f_1(v) \sigma}U_{f_1}\ =\ U_{f_2} \widehat{C}^{\epsilon}_{\gamma v\sigma}.
\end{equation} 
A simple calculation then completes the proof the following 
\begin{prop}
Let $\widehat{C}^*_{\rm Euc}(N)$ be an operator on Cyl$^*$ obtained as a limit 
\begin{equation*}
\widehat{C}^*_{\rm Euc}(N):= \lim_{\epsilon\rightarrow 0} \left(\widehat{C}^{\epsilon}_{\rm Euc}(N)\right)^*, 
\end{equation*}
where $\widehat{C}_{\rm Euc}(N)$ is of the form \eqref{opform} satisfying 
\begin{itemize}    
\item[(a)] covariance  \eqref{cov} under changes of $\epsilon$,
\item[(b)] covariance \eqref{CffC} under Diff$_V$.
\end{itemize}
Then $\widehat{C}^*_{\rm Euc}(N)$ preserves $\mathcal{H}_{\new}$, i.e., 
\begin{equation}\label{inHnew}
\widehat{C}^*_{\rm Euc}(N) \gbra{\Psi}\ \in {\cal H}_{\new},
\end{equation}
and it is diffeomorphism covariant, i.e., 
\begin{equation}
\widehat{C}^*_{\rm Euc}(N)U_f\ =\ U_f\widehat{C}^*_{\rm Euc}(f^{-1*}N). 
\end{equation}
\end{prop}
Above,  the assumption made in \cite{Ashtekar:2004eh} and adopted here, that in (\ref{la}) 
the assigned loops $\alpha_\sigma$ do not overlap $\gamma$ is relevant for (\ref{inHnew}) . 
Otherwise, the operator could produce non-normalizable elements of $\cyl^*$.
We proceed to discuss some further properties of $\widehat{C}^*_{\rm Euc}(N)$ under the assumptions of 
the previous proposition. Firstly, note that we can write 
\begin{equation} 
\widehat{C}^*_{\rm Euc}(N)\ =\ \sum_{x\in \Sigma}N(x)\widehat{C}^*_{{\rm Euc},x}, 
\end{equation}
where $\widehat{C}^*_{{\rm Euc},x}$ has the following properties: It preserves the spaces\footnote{In this paper we are using
a vertex set preserving regularization of the quantum scalar constraint operator introduced in \cite{Ashtekar:2004eh}. In the case of a regularization
which adds  extra vertices, we would have to modify our definition of the space ${\cal H}_V$ suitably. For example, if the 
quantum scalar constraint operator is defined such that it  adds planar vertices \cite{Rovelli:1993}, the new ${\cal H}_V$ would involve  the graphs such that 
the non-planar vertex set is $V$} 
$\eta(\cyl)_V\ :=\ \eta(\cyl) \cap {\cal H}_{V}$ for  $V\in {\rm FS}(\Sigma)$, 
$$\widehat{C}^*_{{\rm Euc},x}  \eta(\cyl)_{V}\subseteq \ \eta(\cyl)_{V} . 
$$
This makes clear that the operator $\widehat{C}^*_{\rm Euc}(N)$ preserves the decomposition \eqref{decomp1} of $\mathcal{H}_{\new}$ into sectors labeled by finite subsets $V$ of $\Sigma$. 
Moreover,  
$$\widehat{C}^*_{{\rm Euc},x}\rvert_{\eta(\cyl)_{V}}\ =\ 0, \ \ \ {\rm unless}\ \ \ x\in V.  
$$
Also, $\widehat{C}^*_{{\rm Euc},x}$ is covariant, 
$$    U_f\widehat{C}^*_{{\rm Euc},x}U_{f^{-1}}\ =\ \widehat{C}^*_{{\rm Euc},{f(x)}}, $$
for every $f\in{\rm Diff}(\Sigma)$. 

Finally,  $\widehat{C}^*_{\rm Euc}(N)$ does not preserve the decomposition \eqref{decomp2}. Rather, 
by the duality to (\ref{epsilon}),
the operator annihilates  the loops created by each $\widehat{C}^\epsilon_{\gamma v \sigma}$. 

The operator $\widehat{C}^*_{\rm Euc}(N)$ is not symmetric. 
But the Hermitian adjoint 
$$ \left(\widehat{C}^*_{\rm Euc}(N)\right)^\dagger $$
is well defined.
 A typical proposal for a symmetric quantum scalar constraint operator is 
\begin{equation}
 \label{sa1}\widehat{C}_{\rm Euc}(N)\ :=\ \frac{1}{2}\left(\widehat{C}^*_{\rm Euc}(N) + \left(\widehat{C}^*_{\rm Euc}(N)\right)^\dagger\right).
\end{equation}
The (essential) self-adjoitness  is an open issue.
%---------------------------------------------------------------------------------------
\subsection{The quantum  Lorentzian scalar constraint of matter-free gravity}
\label{se_lor}
%---------------------------------------------------------------------------------------

To define the quantum scalar constraint operator of the Lorentzian gravity and with a general value of
the Barbero-Immirzi parameter $\beta$, we go back to the classical theory. The gravitational part of the 
scalar constraint is  
$$C(N)\ =\  \sqrt{\beta}C_{\rm Euc}(N)\ -\ 2\frac{(1+\beta^2)}{(8\pi G)^4 \beta^6}T(N) $$
where $T$ is written in a way compatible with the LQG as follows \cite{gr-qc/9606089}       
\begin{equation}
\begin{split}  
T(N)  = {-2}&\int d^3 x \epsilon^{abc} {\rm Tr}
\big(\{A_a(x),\{C_{\rm Euc}(1),V(1)\}\} \cdot\\
&\cdot\{A_a(x),\{C_{\rm Euc}(1),V(1)\}\} \{A_c(x),V(N)\} \big)  
\end{split}
\end{equation}
As before, for every subspace ${\cal H}_\gamma$ in the decomposition (\ref{decomp}) 
we use  the family of paths $s^\epsilon_\sigma$ introduced above,  and a regulated classical expression    
\begin{equation}
\begin{split}
T^\epsilon(N) &= \sum_{\sigma,\sigma',\sigma"} e^{\sigma \sigma' \sigma"}
 {\rm Tr}\Big( H^{-1}_\sigma\{ H_\sigma , K \}   \cdot\\
 &\quad\qquad \cdot H^{-1}_{\sigma'}\{ H_{\sigma'} , K \} 
H_{\sigma''}\{ H^{-1}_{\sigma''} , V(N) \}  \Big)
\end{split}
\end{equation}
with
\begin{equation*}
K\ :=\ \{C_{\rm Euc}(1),V(1)\}, \quad H_{\sigma}:=\rho(h_{s^\epsilon_\sigma})
\end{equation*}   
such that as $\epsilon\rightarrow 0$, the paths are shrunk, the constants $e^{\sigma\sigma'\sigma"}$ are independent of $\epsilon$,   
and\footnote{The path assignment $\sigma\mapsto s^\epsilon_\sigma$ obtained by ignoring the loops $\alpha$ in the assignments $\sigma\mapsto \alpha^\epsilon_\sigma,s^\epsilon_\sigma$ used previously
may assign the same segment to two different $\sigma\not=\sigma'$. That may be compensated
by choosing suitable values for the constants $e^{\sigma\sigma'\sigma"}$.} 
\begin{equation}\label{T->T} 
\lim_{\epsilon\rightarrow 0} T^\epsilon(N)[A,E] = T(N)[A,E] .
\end{equation} 
We introduce that regulation for every graph $\gamma$.  Next, in the kinematical ${\cal H}_{\rm kin}$
we define a quantum operator\footnote{Notice, that in the classical regulated expression $T^\epsilon(N)$ we have $C_{\rm Euc}$ whereas in the quantum regulated expression
we use $\widehat{C}^\epsilon_{\rm Euc}$.  If in the classical expression we replaced      $C_{\rm Euc}$ by $C^\epsilon_{\rm Euc}$, then
(\ref{T->T}) would not be true. On the other hand, we can not use $\widehat{C}^{\rm Euc}$ in the quantum $\widehat{T}^\epsilon(N)$, 
because the expression would not make sense. This is a drawback of the regularization procedure of the Lorentzian part of the scalar constraint. }
$ \widehat{T}^\epsilon(N)\ :\ \cyl \ \rightarrow \ \cyl^*$ via
\begin{equation}
\begin{split}
\widehat{T}^\epsilon(N) &= \frac{1}{(i\hbar)^3}  \sum_{s\in{\cal S}} e^{\sigma \sigma' \sigma"}
 {\rm Tr}\Big( {H}^{-1}_\sigma [ {H}_\sigma , \widehat{K}^\epsilon ]  \cdot\\
 &\qquad\quad \cdot {H}^{-1}_{\sigma'} [{H}_{\sigma'} , \widehat{K} ]
{H}_{\sigma''} [{H}^{-1}_{\sigma''} , \widehat{V}(N)]  \Big)
\end{split}
\end{equation}
with 
\begin{equation*}
\widehat{K}^\epsilon\ :=\   \frac{1}{i\hbar}[\widehat{C}^\epsilon_{\rm Euc}(1),\widehat{V}(1)],
 \end{equation*} 
As in the case of the Euclidean quantum  gravitational constraint,  
$$ \eta({\Psi})\left(\widehat{T}^{\epsilon_1}(N)\Psi' \right)\ =\  \eta({\Psi})\left(\widehat{T}^{\epsilon_2}(N)\Psi' \right) $$  
for $\Psi,\Psi'\in\cyl$, hence the limit is well defined as an operator
$$ \widehat{T}^*(N): \eta(\cyl)\ \rightarrow\ \cyl^*.$$
If the constants $e^{\sigma\sigma'\sigma"}$ are assigned to each graph in a  Diff$(\Sigma)_{\ve(\gamma)}$  invariant way, then analogously to the Euclidean case,
$$ \widehat{T}^*(N) \eta(\cyl)\ \subseteq\ \eta(\cyl), $$
hence  $\widehat{T}^*(N)$ becomes an operator in ${\cal H}_{\new}$ with domain  
$\eta({\rm Cyl})$. 
The operator has a similar structure as $\widehat{C}^*_{\rm Euc}(N)$:
$$\widehat{T}^*(N)\ =\ \sum_{x\in\Sigma}N(x)\widehat{T}^*_x, $$
where, for any $V\in {\rm FS}(\Sigma)$,
$$\widehat{T}^*_x\ :\   {\cal H}_{V}\ \rightarrow \ {\cal H}_{V},  $$
and 
$$\widehat{T}^*_x|_{{\cal H}_{V}}\ =\ 0, \ \ \ {\rm unless}\ \ \ x\in V.  $$
If the constants $e^{\sigma\sigma'\sigma"}$ are assigned to each graph in a Diff$(\Sigma)$-invariant way, then
$$ U_{f}\widehat{T}^*_xU^{-1}_f\ =\ \widehat{T}^*_{f(x)} . $$   
that is the distribution $x\mapsto \widehat{T}^*_x$ is Diff$(\Sigma)$-invariant.

If the operator $\left(\widehat{C}^*_{\rm Euc}\right)^\dagger$ exists, then so does $(\widehat{T}^*(N))^\dagger$. In that case  we can define a symmetric 
operator
\begin{equation}
\label{sa2}
\widehat{T}(N)=\frac{1}{2}\left(  \widehat{T}^*(N) + (\widehat{T}^*(N))^\dagger \right). 
\end{equation}
The final result is a  quantum gravitational scalar constraint operator
\begin{align}  
\widehat{C}(N)\ =\  \sqrt{\beta}\widehat{C}_{\rm Euc}(N)\ -\ 2\frac{(1+\beta^2)}{(8\pi G)^4 \beta^6}\widehat{T}(N)  
 \end{align}   
defined in  ${\cal H}_{\new}$  in the domain $\eta({\rm Cyl})$. As a consequence of the properties 
of $\widehat{C}_{\rm Euc}(N)$ and $T(N)$, it is again local and covariant, 
\begin{align*}
\widehat{C}(N)\ &=\ \sum_{x\in\Sigma}N(x)\widehat{C}_x,\\
\widehat{C}_x {\cal H}_{V}\ &\subseteq \ {\cal H}_{V}\\
\widehat{C}_x\rvert_{{\cal H}_{V}}\ & =\ 0, \ \ \ {\rm unless}\ \ \ x\in V,\\
U_{f}\widehat{C}_xU^{-1}_f\ &=\ \widehat{C}_{f(x)}.
\end{align*}
%-------------------------------------------------------------------
\subsection{Solutions to the quantum constraints.}
\label{se_sol}
%-------------------------------------------------------------------
Suppose the quantum constraint operators $\widehat{C}_x$, $x\in\Sigma$, are essentially self-adjoint. 
Since 
$$ [\widehat{C}_x,\widehat{C}_{x'}]\ =\ 0, $$
every subspace ${\cal H}_{\{x_1,\ldots ,x_m\}}$ can be decomposed using the spectral decomposition
of the operators $\widehat{C}_{x_I}$, $I=1,\ldots ,m$,
$$ {\cal H}_{\{x_1,\ldots ,x_m\}}\ =\ \int^\oplus d\mu(c_1)\ldots d\mu(c_m) {\cal H}_{\{x_1,\ldots ,x_m\}}^{c_1\ldots c_m}. $$
The elements of the subspace 
$$ {\cal H}_{\{x_1,\ldots ,x_m\}}^{0\ldots 0}$$
are solutions to the  quantum scalar constraint. If $(c_1,\ldots ,c_m)=(0,\ldots ,0)$ is a point of the measure zero,
then some continuity in the map
$$ (c_1,\ldots ,c_m)\ \mapsto\ {\cal H}_{\{x_1,\ldots ,x_m\}}^{c_1\ldots c_m}$$
is used to determine individual spaces ${\cal H}_{\{x_1,\ldots ,x_m\}}^{c_1\ldots c_m}$. 
In the general case,  
$${\cal H}_{\{x_1,\ldots ,x_m\}}^{0\ldots 0}\ \subset\ (\eta(\cyl)_{x_1,\ldots ,x_m})^*.$$ 
The elements are (finite or formal infinite) linear combinations 
\begin{equation} \label{H0000}{\cal H}_{V}^{0\ldots 0} \ni \Psi\ =\ \sum_{[\gamma]\in[\gamma(V)]} \eta(\Psi_\gamma) 
\end{equation}
where $[\gamma]$ ranges over the set of Diff$(\Sigma)_V$ equivalence classes of the graphs
with vertices $V$.  In fact, there is a natural embedding 
$$ {\cal H}_{V}^{0\ldots 0}\rightarrow \cyl^* , $$
in $\cyl^*$, the  infinite formal sum (\ref{H0000}) becomes  a well-defined element.
      
To turn elements of   ${\cal H}_{V}^{0\ldots 0}$ into solutions to the quantum diffeomorphism
constraint we average them with respect to the remaining diffeomomorphisms
$$ \Psi\ =\ \sum_{[\gamma]\in[\gamma(V)]} \eta(\Psi_\gamma) \ \mapsto \  \tilde{\eta}(\Psi)\ =\  \sum_{[\gamma]} \sum_{[f]}\eta(U_f\Psi_\gamma)$$ 
where the last sum ranges
\begin{equation*}
 [f]\ \in\ {\rm Diff}(\Sigma)\,/\,  {\rm Diff}(\Sigma)_{\ve(\gamma)}.
\end{equation*}
The result is the subspace
$$ \tilde{\eta}({\cal H}_{\ve(\gamma)}^{0\ldots 0})\ \subset\ \cyl^*, $$
and its elements are Diff$(\Sigma)$-invariant. On the other hand, the operator $\widehat{C}(N)$ we have defined
can be applied directly on each Diff$(\Sigma)$ invariant element of $\cyl^*$. Then, for every $\Psi\in {\cal H}_{V}^{0\ldots 0}$
\begin{align} \widehat{C}(N) \tilde{\eta}(\Psi)\ &=\  \sum_{[\gamma]} \sum_{[f]}\hat{C}(N)\eta(U_f\Psi_\gamma)\nonumber\\  
&=\ \sum_{[\gamma]} \sum_{[f]}U_f\hat{C}(N\circ f)\eta(\Psi_\gamma)\ =\ 0. 
  \end{align}
Solving the Gauss constraint is ensured either by restricting ${\cal H}$ to the Yang-Mills gauge-invariant 
elements, or by introducing a third rigging map, integration with respect to the SU(2) transformations in $V$
for each space  ${\cal H}_{V}^{0\ldots 0}$.\newline

%--------------------------------------------
\section{Summary and outlook}
\label{se_sum}
%--------------------------------------------
In this article, we have introduced a new Hilbert space $\mathcal{H}_{\new}$ of quantum states for the gravitational field. It can be decomposed into sectors
\begin{equation*}
{\cal H}_{\new}\ =\ \overline{\bigoplus_{V\in {\rm FS}(\Sigma)}{\cal H}_{V}}
\end{equation*}
where the states in ${\cal H}_{V}$ are invariant under all the spatial diffeomorphisms that leave invariant the finite set $V$. 

Using the ideas of \cite{gr-qc/9606089,gr-qc/9606090}, together with the class of regularizations introduced in \cite{Ashtekar:2004eh}, we were able to find quantizations of the scalar constraint of pure gravity $\widehat{C}(N)$  as operators leaving ${\cal H}_{\new}$ invariant. This removes a longstanding technical problem, as previous quantizations were 
either defined on the kinematic Hilbert space $\mathcal{H}_{\rm kin}$ without the possibility to directly remove the regulator, or on the Hilbert space of fully diffeomorphism invariant states $\mathcal{H}_{\rm diff}$, which is not left invariant under the action of the constraint with nonconstant lapse.

In our setup, it is straightforward to symmetrize the operator, see (\ref{sa1},\ref{sa2}). Moreover, one can immediately work out the commutation relations. Since 
\begin{gather}
\label{sum}
\widehat{C}(N)\ =\ \sum_{c\in \Sigma}N(x)\widehat{C}_{x}, \\
\label{comm}
[\widehat{C}_{x},\widehat{C}_{x'}]\ =\ 0
\end{gather}
we find 
\begin{equation}
\label{zero}
[\widehat{C}(M), \widehat{C}(N)]\ =\ 0. 
\end{equation}
To discuss the question of anomalies of this quantization, one would thus have to investigate the quantization of the diffeomorphism generator which would classically result from the Poisson bracket of two scalar constraints, as has been done for Thiemann's quantization, \cite{gr-qc/9705017,gr-qc/9710016,Gambini:1997bc}. It is interesting to note that \eqref{zero} immediately results for any quantization of the form \eqref{sum} under the reasonable condition \eqref{comm}. 

We emphasize again that Thiemann also defines a symmetric constraint, albeit, in a sense, at finite regulator \cite{gr-qc/9606090}. A solution space of this constraint can be defined in $\mathcal{H}_{\rm diff}$. It is a very interesting -- and open -- question, how solutions to the new scalar constraint on 
$\mathcal{H}_{\new}$ relate to those of Thiemann's symmetric scalar constraint \cite{gr-qc/9606090}. Similarly, one should symmetrize Thiemann's non-symmetric scalar constraint on $\mathcal{H}_{\rm diff}$ for the case of \emph{constant} lapse function and compare to our proposal.

There is a very interesting different line of thought, \cite{Laddha:2010wp,Laddha:2011mk,Tomlin:2012qz,Varadarajan:2012re}, which also suggests that one should use a different Hilbert space to represent the (diffeomorphism and scalar) constraints. Those methods carry the additional benefit that they address the question of anomalies in a direct fashion. What connection, if any, they have to the constructions of the present article, remains to be seen.

%--------------------------------------------
\begin{acknowledgments}
%--------------------------------------------
This work was partially supported by the grant of the Polish Narodowe Centrum Nauki nr
2011/02/A/ST2/00300, Foundation for Polish Science and by the Emerging Fields Project \emph{Quantum Geometry} of the Friedrich-Alexander University Erlangen-N\"urnberg. JL thanks members of Institute for Quantum Gravity at the Friedrich-Alexander University Erlangen-N\"urnberg where this work was started for hospitality. HS would also like to thank organizers of the 1st Conference of the Polish Society on GR, where part of this work was completed, for hospitality.   

\end{acknowledgments}

\end{document}